\documentclass[a4paper,11pt]{amsart}
\usepackage{amsmath,amssymb,amsthm}
\usepackage[pic,curve,cmtip]{xy}
\usepackage[hidelinks]{hyperref}

\newtheorem{thm}{Theorem}[section]
\newtheorem{prop}[thm]{Proposition}
\newtheorem{lem}[thm]{Lemma}

\theoremstyle{definition}
\newtheorem{defn}[thm]{Definition}
\newtheorem{exa}[thm]{Example}

\theoremstyle{remark}
\newtheorem{rem}[thm]{Remark}

\newcommand{\defeqn}{\mathrel{\mathop:}\Leftrightarrow}
\newcommand{\defeq}{\mathrel{\mathop :}=}
\newcommand{\approxsim}{{/\!\! \approx}}
\newcommand{\lcoset}{{/_{\!\ell\,}}}
\newcommand{\modsim}{{/\!\! \sim}}
\newcommand{\on}{\operatorname}
\newcommand{\F}{\mathbb F}
\newcommand{\Z}{\mathbb Z}
\newcommand{\R}{\mathbb R}
\let\phi\varphi

\newcommand{\simto}
{\xrightarrow{\raisebox{-0.7ex}[0ex][0ex]{$\sim$}}}

\begin{document}

\title{Cryptographic Group and Semigroup Actions}

\author{Oliver W.~Gnilke}
\address{Department of Mathematical Sciences\\
  Aalborg University, Denmark}
\email{owg@math.aau.dk}

\author{Jens Zumbr\"agel}
\address{Faculty of Comp.\ Science and Mathematics\\
  University of Passau, Germany}
\email{jens.zumbraegel@uni-passau.de}

\subjclass[2010]{94A60}

\thanks{The material in this paper was presented in part at the Workshop
    on Coding and Cryptography (WCC 2022) in Rostock, Germany.}

\maketitle

{\centering\footnotesize\emph{Dedicated to Joachim Rosenthal on the
    occasion of his 60th birthday}\\}

\begin{abstract}
  We consider actions of a group or a semigroup on a set, which
  generalize the setup of discrete logarithm based cryptosystems.
  Such cryptographic group actions have gained increasing attention
  recently in the context of isogeny-based cryptography.  We introduce
  generic algorithms for the semigroup action problem and discuss
  lower and upper bounds.  Also, we investigate Pohlig-Hellman type
  attacks in a general sense.  In particular, we consider reductions
  provided by non-invertible elements in a semigroup, and we deal with
  subgroups in the case of group actions. \medskip

  \noindent \textbf{Keywords}: Discrete logarithm problem;
  cryptographic group action; semigroup action problem.
\end{abstract}


\section{Introduction}

The discrete logarithm problem has a long and profound history
(see~\cite{GJ21} for a recent survey).  In cryptography it has been
playing a key role ever since Diffie and Hellman have based the
security of their famous protocol~\cite{DH76} on the hardness of
computing discrete logarithms modulo a large prime~$p$.  The
underlying group~$\F_p^*$ has then been generalized, most notably to
the $\F_q$-rational points on an elliptic curve, due to
Miller~\cite{Mil86} and Koblitz~\cite{Kob87}.  In fact, while the
discrete logarithm problem in the unit group~$\F_q^*$ of a finite
field can be solved in subexponential time by index calculus
algorithms (for an overview, see~\cite{GKZ18}), the fastest known
algorithm in a general elliptic curve is basically a generic one that
requires exponential time.

However, Shor's quantum algorithm~\cite{Sho97b} constitutes a
polynomial time attack on the discrete logarithm problem in any group
(as well as on the integer factorization problem).  These
observations, and reports on the progress in building quantum
computers now achieving a “quantum supremacy”~\cite{Aru+19}, underline
the need for new concepts to build cryptosystems resistant to quantum
attacks.  An interesting approach is based on isogenies of
supersingular elliptic curves (SIDH,~\cite{DJP14}), which however
is broken due to the Castryck-Decru attack~\cite{CD22}.

More recently, a commutative supersingular isogeny-based
Diffie-Hellman scheme (CSIDH,~\cite{Cas+18}) has been proposed as a
more efficient variant, which is based on the action of the class
group of an endomorphism ring on isomorphism classes of elliptic
curves.  This is an example for an action of an abelian group on a
set, which as a framework suffices to build a Diffie-Hellman protocol,
as has been observed by Couveignes~\cite{Cou06} and independently by
Rostovtsev and Stolbunov~\cite{RS06,Sto10}.

With a somewhat different background, Maze, Monico and Rosenthal have
introduced actions of commutative semigroups on
sets~\cite{Mon02,MMR07} in order to generalize the discrete logarithm
problem, one motivation being to find examples that do not allow even
an (exponential) square-root attack.  The setup was further
investigated by the theses of the present authors, in which
Zumbrägel~\cite{Zum08} considered a generalization to non-commutative
semigroups and Gnilke~\cite{Gni14} showed how a Pohlig-Hellman like
reduction applies to semigroups with non-invertible elements.

In this work, we revisit the concept of semigroup actions for discrete
logarithm based cryptosystems and connect it to recent proposals of
isogeny-based cryptography.  In the case of abelian group actions,
with a view towards isogenies, this has been considered by
Couveignes~\cite{Cou06} and, more recently, by Smith~\cite{Smi18} and
Alamati et al.~\cite{ADMP20}.  Here we aim to take a slightly more
abstract viewpoint and introduce generic algorithms for the semigroup
action problem.  We also investigate Pohlig-Hellman type reductions
and consider those provided by non-units in a semigroup and by
subgroups in the context of a group action.


\section{Cryptographic semigroup actions}

In this section we briefly recall the notion of a semigroup action and
its application to cryptography~\cite{Mon02,MMR07}.

By a \emph{semigroup} we mean a set~$S$ with an associative binary
operation (written multiplicatively).  It is called a \emph{monoid} if
a neutral element exists. A \emph{semigroup action} with respect to a
semigroup~$S$ and a set~$X$ is a map
\[ S \times X \to X \,, \quad (s, x) \mapsto s . x \]
such that $s t . x = s . (t . x)$ for all $s, t \in S$ and $x \in X$.
Considering for $s \in S$ the transformation $\Phi_s \colon X \to X$,
$x \mapsto s . x$, this means that $\Phi_{s t} = \Phi_s \circ \Phi_t$
for $s, t \in S$.  When there is a semigroup action, then~$X$ is also
called \emph{$S$-set}.

\begin{defn}
  Consider a semigroup~$S$ acting on a set~$X$.  The \emph{semigroup
    action problem} is the problem, for given $x, y \in X$ to find
  some $s \in S$ such that $y = s . x$.
\end{defn}

So the above problem asks to find preimages of the “orbit map”
\[ \Psi_x \colon S \to X \,, \quad s \mapsto s . x \,. \]
Suppose that~$S$ is a group and let $S_x = \{ s \in S \mid s . x = x \}$
be the \emph{stabilizer subgroup} of $x \in X$.  Then the orbit map
induces a bijection $S \lcoset S_x \simto S . x$ of the left cosets
and the orbit.  Thus solutions to the semigroup action problem (which
we also call \emph{group action problem}) are unique up to left
congruence modulo the stabilizer.

Since we deal with cryptographic applications, we assume all
structures to be finite and that the semigroup action is efficiently
computable.  This means that elements of the semigroup~$S$ and the
set~$X$ are encoded by bit strings, and both the semigroup operation
and the action are computable in polynomial time.

\begin{exa}\label{exa:groupexp}
  Consider a finite cyclic group $(G, \cdot)$ of order~$n$, which may
  be seen as a $\Z_n$-module.  If we “forget” its additive structure,
  we have \[ (\Z_n, \cdot) \times G \to G \,, \quad (s, g)
    \mapsto g^s \,, \] and the semigroup action problem is just
  the discrete logarithm problem in the group~$G$.  Note that the
  action is efficiently computable by a square-and-multiply method,
  provided the group operation is efficient.
\end{exa}

In order to set up a Diffie-Hellman like key agreement, we need some
way to generate commuting elements of the semigroup~$S$.  For
simplicity we assume here the semigroup to be commutative and
have the following key agreement scheme:
\begin{center}
  \begin{tabular}{ccccc}
    Alice && \emph{public} && Bob \\\hline
          && $x \in X$ && \\
    $a \in S$ &$\rightarrow$~~~~& $a . x \in X$ && \\
          && $b . x \in X$ &~~~~$\leftarrow$& $b \in S$ \\
    $k_A = a . (b . x)$ && && $k_B = b . (a. x)$
  \end{tabular}
\end{center}
Observe that both parties compute the same key $a b . x = b a . x$.
Also notice that in the case of Example~\ref{exa:groupexp} the scheme
amounts to classical Diffie-Hellman.

\begin{defn}
  Consider a commutative semigroup~$S$ acting on a set~$X$.  The
  \emph{semigroup Diffie-Hellman problem} is the problem, for given
  $x, y, z \in X$ to find some $k \in X$ such that $y = a . x$,
  $z = b . x$ and $k = a b . x = b a . x$ for some $a, b \in S$.
\end{defn}

It is clear that if one can solve the semigroup action problem, then
one can break the Diffie-Hellman protocol, while the converse
direction is not obvious.  There have been results on subexponential
reductions in the classical case of group exponentiation%
~\cite{Mau94,MW99}, and recently on polynomial quantum reduction for
abelian group actions~\cite{GPSV21,MZ22}.

\begin{exa}\label{exa:csidh}
  In isogeny-based cryptography and CSIDH, a particular case of
  interest is the action of an abelian group (the class group in an
  endomorphism ring) on a set~$X$ (of isomorphism classes of elliptic
  curves), cf.~\cite{Cou06,Cas+18,Smi18}.  In Couveignes' work~%
  \cite{Cou06} the group action is also assumed to be simply
  transitive, and the semigroup action problem and the semigroup
  Diffie-Hellman problem are called \emph{vectorization problem} and
  \emph{parallelization problem}, respectively; if those are
  intractable, the set~$X$ is referred to as a \emph{hard
    homogeneous space}.

  It would be interesting to view SIDH also in the framework of a
  group action, but this seems not to be obvious, cf.~\cite[Sec.~15]%
  {Smi18}.
\end{exa}


\section{Generic algorithms}

We use Maurer's abstract model of computation~\cite{Mau05} to describe
generic algorithms for a semigroup action.  Recall that in this model
one specifies
\begin{itemize}
\item a ground set~$M$,
\item a set~$\Pi$ of certain operations $f \colon M^t \to M$ of
  arity $t \in \{ 0, 1, 2, \dots \}$,
\item a set~$\Sigma$ of certain relations $\rho \subseteq M^t$ of
  arity $t \in \{ 1, 2, \dots \}$.
\end{itemize}
There are internal state variables $V_1, V_2, \dots$ storing elements
in~$M$, which cannot be read directly.  A \emph{generic algorithm}
$A(M, \Pi, \Sigma)$ is then allowed to perform computation operations
and relation queries for $f \in \Pi$ and $\rho \in \Sigma$, using the
internal state variables as input (and operation output).  In the
analysis we usually discard the relation queries and only count the
number of operations performed.

\begin{exa}
  Generic algorithms in a cyclic group of order~$n$ can be modeled
  using $M = \Z_n$, $\Pi = \{ + \}$ and $\Sigma = \{ = \}$.  For
  the discrete logarithm problem the state variables are initialized
  with $V_1 = 1$ and $V_2 = s \in \Z_n$ random, and the goal is to
  solve the \emph{extraction problem} for the value $s \in V_2$.
\end{exa}

Now let $S \times X \to X$ be a semigroup action and fix $x \in X$.
In order to address the semigroup action problem, we define the ground
set as
\[ M \defeq S \modsim_x \,, \quad\text{where}\quad
  s \sim_x t ~\defeqn~ s . x = t . x \,, \]
and denote the class of $s \in S$ by $[s]$.  We allow for any $a \in S$
to perform the unary operation $\Phi_a \colon M \to M$, $[s] \mapsto
a . [s] \defeq [a s]$.  The state variables are initialized with
$V_1 = [1]$ (assuming that $1 \in S$, otherwise we adjoin it) and
$V_2 = [s] \in M$, which describes the solutions~$s$ to a semigroup
action problem $y = s . x$.  The goal is to solve the extraction
problem for~$V_2$, i.e., to find some $s' \in S$ such that
$[s] = [s']$.

If we are dealing with a group action, the set $S \modsim_x = S
\lcoset S_x$ consists of the left cosets $[s] = s S_x$ of the
stabilizer subgroup of~$x$.  In this case, we can show the following
result, which provides a generic lower bound of $\Omega(\sqrt n)$ for
the group action problem in a set~$X$ of size~$n$.

\begin{thm}\label{thm:generic}
  Let $S \times X \to X$ be a transitive group action, which is
  abelian or free.  Fix $x \in X$ and let $M \defeq S \lcoset S_x$,
  $\Pi = \{ \Phi_a \mid a \in S \}$, $\Sigma = \{ = \}$ as above.
  If $s \in S$ is uniformly random, the success probability of a
  generic algorithm $A(M, \Pi, \Sigma)$ for the group action problem
  $y = s . x$ using~$m$ operation queries is at most $\frac 1 4 m^2
  / \vert X \vert$.
\end{thm}

\begin{proof}
  Following the proof of~\cite[Thm.~1]{Mau05} it suffices to upper
  bound the probability that a collision in the state variables
  occurs.  These entries are either of the form $a . [1] = [a]$
  or of the form $b . [s]$, for known $a, b \in S$.  A collision
  of type $[a] = [a']$ is independent of~$s$ and can be discarded.
  Moreover, a collision of type $b . [s] = b' . [s]$ means that
  $b^{-1} b' \in S_{s . x} = s S_x s^{-1}$, but since the action is
  abelian or free we have $s S_x s^{-1} = S_x$, again independent
  of~$s$.

  This means that the only collisions related to~$s$ are of the
  form \[ a . [1] = b . [s] \] for some $a, b \in S$.  Such an
  event is equivalent to $[s] = [b^{-1} a]$, and since $[s] \in M$
  is uniformly random it occurs with probability $1 / n$, where
  $n \defeq \vert M \vert = \vert X \vert$.  So if the algorithm
  computes~$u$ and~$v$ values of the form
  $a . [1]$ and $b . [s]$ respectively, the
  probability for a collision is at most $u v / n$.  But since
  $u \!+\! v \le m$ we have $u v \le \frac 1 4 m^2$, from which the
  result follows.
\end{proof}

The following example shows that there may be much faster generic
algorithms in case the group action is neither abelian nor free.

\begin{exa}
  Let~$X$ be a set of size~$n$ and let the symmetric group $S \defeq
  \on{Sym} X$ act on~$X$.  Fix $x \in X$ and let $s \in S$ be random.
  A generic algorithm to find the coset $[s] \in S \lcoset S_x$ can
  employ either “usual” collisions $[a] = b . [s]$, from which a
  solution $s' \defeq b^{-1} a$ can be obtained, or collisions of the
  form \[ b . [s] = b' . [s] , \] i.e., $c . [s] = [s]$ for
  $c \defeq b^{-1} b'$.  If we choose $c \in S$ having~$k$ fixed
  points, then the event probability is $k / n$, and thus we may apply
  a divide-and-conquer strategy to obtain $[s]$ in $O(\log n)$ steps.
\end{exa}

Furthermore, for proper semigroup actions the difficulty of the
generic semigroup action problem very much depends on the structure of
the semigroup, and ranges from efficient algorithms in $O(\log n)$ to
lower bounds of $\Omega(n)$, see the examples below.  From a
cryptography perspective there are however issues with applying those
actions, as discussed in Section~\ref{sec:pohlig-hellman}.

\begin{exa}\label{exa:generic}
  Let~$S$ be a semigroup, $X$~a set and $\phi \colon S \to X$ a
  bijection.  Then we can make~$X$ an $S$-set by letting  
  \[ s . \phi(t) = \phi(s t) \]
  for $s, t \in S$.  We assume this action to be efficiently
  computable and think of the inverse map~$\phi^{-1}$ as “hidden”.
  For example, if~$G$ is a cyclic group with generator~$g$, we use the
  bijection $\phi \colon (\Z_n, \cdot) \to G$, $s \mapsto g^s$.  Let
  us look at two further cases.
  \begin{enumerate}
  \item Suppose that $(S, \cdot) = (\{ 1, \dots, n \}, \min)$ and we
    are given a semigroup action problem instance $x, y \in X$ where
    $y = s . x$.  Let $a = \phi^{-1}(x)$ which we suppose is known,
    e.g., $a = n$ if $x = \phi(n)$.  Since $y = s . x = s . \phi(a)
    = \phi(s a) = \phi(s a a) = s a . \phi(a) = s a . x$, we may
    assume that $s = s a$, i.e., $s \le a$.  Then for any $t \in S$,
    $t \le a$, there holds
    \begin{align*} t \le s
      ~~&\Leftrightarrow~~ t = t s ~~\Leftrightarrow~~
          t a = t s a ~~\Leftrightarrow~~
          \phi(t a) = \phi(t s a) \\
      ~~&\Leftrightarrow~~ t . \phi(a) = t . (s . \phi(a))
          ~~\Leftrightarrow~~ t . x = t. y \,.
    \end{align*}
    Hence, we can find~$s$ using binary search in $O(\log n)$ steps.
    \medskip
  \item On the other hand, define $(S, \cdot) = (\{ 0, s_1, \dots,
    s_m, 1 \}, \wedge)$, where~$\wedge$ is a semilattice operation
    such that \[ 0 \wedge s_i = 0 \,, \quad 1 \wedge s_i = s_i \,,
      \quad \text{and} \quad s_i \wedge s_j = 0 \text{ whenever }
      i \ne j \,, \]
    and let $o = \phi(0)$ and $e = \phi(1)$ in~$X$.  Consider a
    semigroup action problem instance $e, y \in X$ where $y = s . e$.
    As $s . e = s . \phi(1) = \phi(s 1) = \phi(s)$ and~$\phi$ is
    bijective, there is a unique solution~$s$.  When using a generic
    algorithm we may for $t \in S$ compute $t . e = \phi(t)$ and
    $t . y = t . (s . e) = t s . e = \phi(t s)$, where a collision
    occurs only if $t = t s$, which if $s = s_i$ means that $t \in
    \{ 0, s_i \}$.  So it requires $\Omega(n)$ steps to find any
    collision and thus information about~$s$.

    Note that this example is not interesting for a Diffie-Hellman
    type key agreement, because the key $k = a . (b . e) = b . (a . e)$
    will usually be~$o$.
  \end{enumerate}
\end{exa} \pagebreak

We should mention here that Shoup's model for generic algorithms~%
\cite{Sho97a} is also widely used, which is based on representations
of algebraic objects as random bitstrings.  In the case of a cyclic
group of order~$n$ one has an injective “encoding function”
$\sigma \colon \Z_n \to \{ 0, 1 \}^*$, and the algorithm maintains a
list of encodings $(\sigma(x_1), \dots, \sigma(x_k))$ to which entries
$\sigma(x_i \!+\! x_j)$ may be appended by performing oracle queries.
The bitstring representation allows a generic algorithm to perform
additional operations, like sorting or hashing of algebraic objects,
and thus enables algorithms based on pseudorandom functions such as
Pollard's rho method.  It is possible to show a square-root lower
bound for the group action problem in Shoup's model by adapting the
proof of Theorem~\ref{thm:generic}.

Regarding upper bounds for the complexity of the semigroup action
problem, generic algorithms often aim at finding a collision in
square-root time.  However, in order to deduce a solution from it one
needs the ability to invert some elements in the semigroup.  Below we
present a version of Shanks' baby-step-giant-step method as well as a
Pollard-rho type attack for the case of group actions.

Note that in the quantum world generic attacks on group actions may be
faster than square-root time.  Indeed, there are subexponential
algorithms solving the group action problem for (free) abelian group
actions based on Kuperberg's quantum algorithm, cf.~\cite[Sec.~7]{MZ22}.


\section{Collision attacks}

Let~$X$ be an $S$-set of size~$n$ and let $x \in X$, $y \in S . x$.
We now discuss generic upper bounds for the corresponding semigroup
action problem, i.e., to find $s \in S$ such that $y = s . x$.  An
important class of generic algorithms relies on finding a collision
$a . x = b . y$ for some $a, b \in S$, from which we can deduce a
solution $s \defeq b^{-1} a$, provided that~$b$ is invertible.

For analyzing such algorithms, the following combinatorial result on
the probability that two random subsets are disjoint is useful
(cf.~\cite[Lem.~68]{Gni14}).

\begin{lem}\label{lem:intersect}
  Let~$X$ be a set of size~$n$ and let $A, B \subseteq X$ be uniformly
  chosen random subsets of size~$k, \ell$ respectively.  Then
  \[ 1 - \tfrac{k \ell} n \,\le\, \Pr( A \cap B = \varnothing )
    \,\le\, \exp({- \tfrac{k \ell} n}) \,. \]
\end{lem}

\begin{proof}
  Let us assume w.\,l.\,o.\,g.\ that~$B$ is fixed.  For the lower
  bound we have
  \[ \Pr(A \cap B \ne \varnothing) = \Pr( \textstyle\bigcup_{x \in B}
    \{ x \in A \}) \le \textstyle\sum_{x \in B} \Pr(x \in A) = \ell
    \!\cdot\! \tfrac k n \]
  by the union bound.  For the upper bound we may assume that~$A$
  consists of~$k$ randomly chosen elements of~$X$, possibly with
  repetition, by which the probability of being disjoint to~$B$ only
  can increase.  Let~$T_i$ be the event that the $i$-th element of~$A$
  does not lie in~$B$, for $1 \le i \le k$.  Then these events are
  independent of probability $1 \!-\! \frac {\ell} n$, whence
  \[ \Pr(A \cap B = \varnothing) \le \Pr( \textstyle\bigcap_i
    T_i ) = (1 \!-\! \frac {\ell} n)^k \le
    \exp({- \frac{\ell k} n}) \,, \]
  since $1 \!+\! x \le \exp x$ holds for all $x \in \R$.
\end{proof}

\subsection*{A time-memory trade-off}

We first present a simple time-memory trade-off attack adopting
Shanks' baby-step-giant-step method (cf.~\cite[Sec.~4.2]{Gni14}).
Let the semigroup~$S$ be a monoid with group of invertible
elements~$S^*$.  The algorithm is described as follows.
\begin{enumerate}
\item precompute a table of entries $(b_j . y, b_j)$ for $b_j \in S^*$
  random
\item for $a \in S$ random, check if $a . x = b_j . y$ for some~$j$,
  if yes output $s \defeq b_j^{-1} a$
\end{enumerate}

For the analysis, according to the proof of Lemma~\ref{lem:intersect}
we may bound the collision probability even if the set $B \defeq \{
b_1 . y, b_2 . y, \dots \}$ is of a special nature (as the $b_j$ have
to be invertible).  It is enough to require that the set $A \defeq
\{ a_1 . x, a_2 . x, \dots \}$ generated by the algorithm behaves as
random, which holds if the orbit map $\Psi_x \colon S \to X$,
$s \mapsto s . x$ has (nearly) constant-sized preimages.

Suppose the algorithm generates $k \defeq \vert A \vert$ elements
$a_i . x$ with $a_i \in S$ and $\ell \defeq \vert B \vert$ elements
$b_j . y$ with $b_j \in S^*$, which is upper bounded by $\vert S^*
\vert$.  Then we infer from Lemma~\ref{lem:intersect} that the success
probability of the algorithm is at least $\frac 1 2$ provided that
$\exp({- \frac{k \ell} n}) \le \frac 1 2$, or equivalently,
\[ k \ell \,\ge\, n \ln 2 \,. \]

Therefore, an optimal choice of parameters is $k, \ell \in
\Theta(\sqrt n)$, which is possible if a sufficient amount of
invertible elements is available and there is enough memory.  In such
a case the complexity of this algorithm is seen to be $O(\sqrt n)$.
Next we show an approach how to drastically reduce the memory
requirement.

\subsection*{A Pollard-rho attack for group actions}

Suppose now that a group~$G$ acts transitively on a set~$X$ of
size~$n$.  Given $x, y \in X$ the group action problem asks to find
$g \in G$ such that $y = g . x$.  We describe a Pollard-rho type
birthday attack for solving this problem, which is a slight adaption
and simplification of an algorithm given by Monico~\cite[Sec.~4.2]%
{Mon02}.

As before, the idea is to generate elements $a_i . x ,\, b_j . y \in X$
in a pseudorandom way and to provoke a collision $a_i . x = b_j . y$
from which a solution $g \defeq b_j^{-1} a_i$ to the group action
problem is deduced.  But here, this is to be done using very little
memory while still maintaining a heuristic square-root complexity
$O(\sqrt n)$.

The algorithm depends on a pseudorandom function $f \colon X \to G$.
Define a recursive sequence $(a_1, a_2, a_3, \dots)$ in~$G$ by
\[ a_1 \defeq a \in G ~\text{random} \,, \qquad
  a_{i+1} \defeq f(a_i . x) \, a_i \quad \text{for $i \ge 1$} \,, \]
hence $a_2 = f(a . x) \, a$, $a_3 = f(f(a . x) \, a . x) \, f(a . x)
\, a$, etc.  Observe that if $a_i . x = a_j . x$ for some~$i, j$ then
$a_{i+r} . x = a_{j+r} . x$ for all~$r$.  Similarly, define a
recursive sequence $(b_1, b_2, b_3, \dots)$ in~$G$ by
\[ b_1 \defeq b \in G ~\text{random} \,, \qquad
  b_{j+1} \defeq f(b_j . y) \, b_j \quad \text{for $j \ge 1$} \,, \]
and note that if $a_i . x = b_j . y$ then $a_{i+r} . x = b_{j+r} . y$
for all~$r$.  The algorithm now consists of two steps:
\begin{enumerate}
\item construct $a$-loop, i.e., find smallest~$k$ such that
  $a_k . x = a_{2k} . x$ \\ construct $b$-loop, i.e., find
  smallest~$\ell$ such that $b_{\ell} . y = b_{2\ell} . y$
\item find a collision in the sets
  \[ A \defeq \{ a_{k+i} . x \mid 0 \le i < k \} \quad\text{and}
    \quad B \defeq \{ b_{\ell+j} . y \mid 0 \le j < \ell \} \]
  by checking $a_k . x = b_{\ell+s} . y$ for $s = 0, 1, 2, \dots$;
  if successful output $g \defeq b_{\ell+s}^{-1} a_k$
\end{enumerate}
Notice here that if $a_{k+i} . x = b_{\ell+j} . y$ for some $i, j$, then
$a_k . x = a_{2k} . x = b_{\ell+s} . y$ for $s \defeq j \!+\! k \!-\! i$.

We sketch a heuristic analysis of this algorithm.  With good
probability we have $\vert A \vert = k$ and $\vert B \vert = \ell$ (if
the preperiod does not exceed the period) and we expect $k, \ell
\in O(\sqrt n)$ by the birthday paradox.  Assuming that $A, B$ behave
as random subsets of~$X$ we can estimate the probability of a
collision using Lemma~\ref{lem:intersect} by \[ \Pr(A \cap B \ne
\varnothing) \ge 1 \!-\! \exp \big( {- \tfrac{k \ell} n} \big) 
\in \Omega(1) \,, \] provided that $k, \ell \in \Omega(\sqrt n)$.
These arguments show that the algorithm has an expected running time
of $O(\sqrt n)$ and succeeds with non-negligible probability.

This algorithm may be adapted to work also for proper semigroup
actions, in case the semigroup~$S$ has sufficiently many invertible
elements~$S^*$ and we are able to define a sequence $(b_1, b_2, b_3,
\dots)$ in~$S^*$.  Alternatively, as pointed out by Maze~\cite{Maz03},
semigroups with a large subgroup can be attacked by excluding all
non-units first and then employing the above algorithm on the unit
subgroup.


\section{Pohlig-Hellman type reductions}%
\label{sec:pohlig-hellman}

In this section we examine how the hardness of the semigroup action
problem is affected by exploiting certain substructures.  We recollect
the framework of Pohlig-Hellman type reductions from~\cite[Sec.~4.3]%
{Gni14} and discuss a few special cases for cryptographic group
actions.

Recall that the classical Pohlig-Hellman algorithm~\cite{PH78}
essentially reduces the difficulty of the discrete logarithm problem
in a cyclic group~$G$ of order~$n$ to that in a group of order the
largest prime factor of~$n$.  The algorithm can be viewed as applying
multiplication-by-$m$ maps \[ \lambda_m \colon \Z_n \to \Z_n \,,
  \quad s \mapsto m s \] in order to reduce the problem to the
action of the (smaller) ideal $m \Z_n$.  The following general concept
captures this scenario and many others.

\begin{defn}\label{def:reduction}
  Let~$S$ and~$T$ be semigroups, let~$X$ be an $S$-set and~$Y$ be a
  $T$-set.  A \emph{reduction} $(f, F, G)$ consists of maps
  $f \colon S \to T$ and $F, G \colon X \to Y$ such that
  \[ f(s) . G(x) = F(s . x) \] for all $s \in S$ and $x \in X$,
  see the diagram below.
  \[ \xymatrix@R=15mm@C=2cm{
      S \ar[r]_{\Psi_x}="x" \ar_f[d] & X \ar^F[d] \\
      T \ar[r]^{\Psi_{G(x)}}="Gx" & Y \ar^G"x";"Gx" } \]
\end{defn}

For a general reduction $(f, F, G)$ the map $f \colon S \to T$ is
not required to be a semigroup homomorphism.  However, if $F = G$
it is reasonable to assume this, since $f(s t) . F(x) = F(s t . x)
= F(s . t . x) = f(s) . F(t . x) = f(s) . f(t) . F(x) = f(s) f(t)
. F(x)$ for all $s, t \in S$ and $x \in X$.

\begin{exa}
  Let~$G$ be a cyclic group with generator~$\alpha$ of composite order
  $n = k m$.  In the Pohlig-Hellman setup above we may apply the
  isomorphism $m \Z_n \cong \Z_k$ after the multiplication-by-$m$
  map~$\lambda_m$, which results in the natural map $\pi \colon \Z_n
  \to \Z_k$.  Then one has a reduction $(\pi, \Phi_m, \Phi_m)$, where
  $\Phi_m \colon \langle \alpha \rangle \to \langle \alpha^m \rangle$,
  $g \mapsto g^m$.  Indeed, there holds $(g^m)^{\pi(s)} = (g^s)^m$ for
  any $s \in \Z_n$ and $g \in G$, see below.
  \[ \xymatrix@R=15mm@C=2cm{
      \Z_n \ar[r]_{\Psi_g}="x" \ar_{\pi}[d] &
      \langle \alpha \rangle \ar^{\Phi_m}[d] \\
      \Z_k \ar[r]^{\Psi_{g^m}}="Gx" &
      \langle \alpha^m \rangle \ar^{\Phi_m}"x";"Gx" } \]
\end{exa}

Given any reduction $(f, F, G)$, an adversary who can solve the
semigroup action problem in~$T$ can restrict the search in~$S$ to
preimages of the solutions in~$T$ under the map~$f$.  Indeed, given a
semigroup action problem instance $x, y \in X$ where $y = s . x$, one
reduces it to the instance $G(x), F(y) \in Y$ where $F(y) = f(s) .
G(x)$.  Nevertheless one should note the following caveats:
\begin{enumerate}
\item In general, for a single solution $t \in \on{im} f \subseteq T$
  of $F(y) = t . G(x)$ it is not clear whether there always exists
  $s \in f^{-1}(t)$ such that $y = s . x$.
\item If the above solution $t \in \on{im} f$ is unique, then $t =
  f(s)$ and we can deduce that $s \in f^{-1}(t)$, which may be a much
  smaller set than~$S$.  However, the preimage under~$f$ could be hard
  to obtain, and often it does not admit a useful semigroup action
  structure itself.
\end{enumerate}  

Similarly to the Pohlig-Hellman algorithm we may apply several
reductions $(f_1, F_1, G_1), \dots, (f_r, F_r, G_r)$ in parallel
to further narrow down the search space to $s \in f_1^{-1}(t_1)
\cap \dots \cap f_r^{-1}(t_r)$ if suitable solutions $t_1, \dots,
t_r$ of the reduced semigroup action problems are found.

\subsection*{Examples based on non-units}

We call a reduction \emph{effective} if the maps $f, F, G$ are
efficiently computable and there holds $1 < \vert T \vert < \vert
S \vert$.  The next result describes a very general class of
reductions.

\begin{prop}
  Let~$S$ be a monoid, $X$ an $S$-set and $m \in S$.  Then the
  triple $(\lambda_m, \Phi_m, \on{id})$ forms a reduction, which is
  effective iff~$m$ is not left-absorbing and not invertible, provided
  the semigroup operation and action are efficient.
\end{prop}

\begin{proof}
  For all $x \in X$ and $s \in S$ there holds that
  \[ \lambda_m(s) . x = m s . x = m . (s . x) = \Phi_m(s . x) \,. \]
  The reduction maps the semigroup~$S$ onto its right ideal $m S$,
  which is a non-trivial proper subsemigroup of~$S$ iff $\lambda_m$
  is not constant or surjective.
\end{proof}

Akin to the Pohlig-Hellman approach for groups of prime power order,
we may apply this reduction recursively.  The idea for solving the
semigroup action problem $y = s . x$ is to find $t_1 \defeq m_1 s$
from $m_1 . y = t_1 . x$ (bearing in mind the caveats above), and for
obtaining~$t_1$ to find $t_2 \defeq m_2 m_1 s$ from $m_2 m_1 . y
= t_2 . x$ etc., for suitable non-units $m_1, m_2, \dots$ of~$S$.

Hence, the result constitutes a considerable threat to the security of
a cryptosystem based on proper semigroup actions.  On the other hand,
the rather degenerate semigroup~$S$ of Example~\ref{exa:generic} (2)
contains many non-absorbing non-units and yet has a generic complexity
of $\Omega(n)$ for the semigroup action problem.  In this case, one
has $\vert m S \vert = 2$ for all non-absorbing non-units~$m$ and the
corresponding preimage sets are not useful (see also \cite[Ex.~79]%
{Gni14}).

In conclusion, while it is conceivable that certain semigroup actions
avoid the attacks outlined in this section and are in fact interesting
for cryptography, possible candidates have to be chosen very
carefully.

\subsection*{Examples based on automorphisms}

Now we consider a second family of reductions, which also applies to
group actions.  Let us start with a general semigroup~$S$ and an
$S$-set~$X$.  Its automorphisms are the bijective maps $\phi \colon
X \to X$ such that $\phi(s . x) = s .  \phi(x)$ for all $s \in S$,
$x \in X$.  We use automorphism groups to construct equivalence
relations on~$X$ compatible with the action as follows.  Suppose that
a group~$H$ acts on~$X$ by automorphisms, i.e., we have $H \times
X \to X$, $(h, x) \mapsto h . x$, such that $h . s . x = s . h. x$ for
$h \in H$, $s \in S$, $x \in X$.  Let $X \modsim$ be the set of its
orbits $[x] = \{ h . x \mid h \in H \}$.  This induces an action
\[ S \times X \modsim \to X \modsim \,, \quad
  (s, [x]) \mapsto [s . x] \,. \]

Suppose next that the semigroup~$S$ is a commutative monoid with unit
group~$S^*$.  Then each $h \in S^*$ defines an $S$-automorphism
$\Psi_h \colon X \to X$ by $x \mapsto h . x$, so we may consider any
subgroup~$H$ of~$S^*$ as a group of automorphisms of~$X$.  The
relation $s \approx t$ iff $s H = t H$ then provides a semigroup
congruence on~$S$, and we denote by $S / H \defeq \{ s H \mid s
\in S \}$ its classes.  Moreover, $s \approx t$ implies
$s . x \sim t . x$, for any $s, t \in S$ and $x \in X$, so we can
define an action \[ S / H \times X \modsim \to X \modsim \,, \quad
  ([s]_{\approx}, [x]_{\sim}) \mapsto [s . x]_{\sim} \,. \]
We hence obtain a reduction $(f, F, F)$ with $f \colon S \to
S \approxsim$ and $F \colon X \to X \modsim$ being the natural maps,
see below.
\[ \xymatrix@R=15mm@C=2cm{
    S \ar[r]_{\Psi_x}="x" \ar_f[d] & X \ar^F[d] \\
    S/H \ar[r]^{\Psi_{[x]}}="Gx" & X \modsim \ar^F"x";"Gx" } \]
In particular, if $S = G$ is an abelian group, so that $S^* = G$, we
can employ any subgroup~$H$ of~$G$ and thus $X \modsim$ becomes a
$G/H$-set.  Regarding the practical implications however, the reduced
action may not be efficiently computable, because the equality of
orbits could be difficult to check.  In any case, an effective
reduction attack cannot be of generic type, since we are guaranteed
a lower square-root complexity for such algorithms due to
Theorem~\ref{thm:generic}.

This approach was also described in the context of CSIDH, cf.~%
\cite[Sec.~12]{Smi18}.

\begin{rem}
  As in Example~\ref{exa:groupexp} consider the discrete logarithm
  setup of an abelian group~$X$ of order~$n$, i.e., we have the
  action of $(\Z_n, \cdot)$ on~$X$ by exponentiation.  Every group
  automorphism is also an automorphism of~$X$ as an $\Z_n$-set, thus
  any subgroup~$H$ of the automorphism group of~$X$ induces an action
  \[ (\Z_n, \cdot) / H \times X \modsim \to X \modsim ,\, \quad
    ([s], [x]) \mapsto [x^s] \,. \]
  In the special case of $H = \{ \pm 1 \}$ we have the orbits
  $[x] = \{ x, x^{-1} \}$, reflecting the practice in elliptic curve
  cryptography to identify the points~$\pm P$ and thus use only
  the $x$-coordinate~\cite{Mil86}.

  Such reductions could potentially weaken the security of the
  discrete logarithm problem in groups for which the automorphism
  group has several subgroups, e.g., in cyclic groups of order~$n$
  where $\phi(n) = \vert \Z_n^* \vert$ is smooth.  However, as
  mentioned above these reductions appear to be not effective in
  general.
\end{rem}

Let us state a concrete example illustrating this phenomenon.

\begin{exa}
  Consider the discrete logarithm problem in a group~$X$ of prime
  order $n = 29$ (e.g., a subgroup of~$\Z_{59}^*$).  Then we have
  $S = (\Z_{29}, \cdot)$ and hence
  \[ S^* = \Z_{29}^* \cong \Z_{28} \cong \Z_4 \!\times\! \Z_7 \,. \]
  Therefore, one could try to exploit the subgroups $H_1 = \langle
  \alpha_1 \rangle$, $H_2 = \langle \alpha_2 \rangle$ of order~$4$
  and~$7$ respectively, say $\alpha_1 = 12 \in \Z_{29}^*$ and
  $\alpha_2 = 7 \in \Z_{29}^*$, to attack to problem.  But the reduced
  group actions seem to be more difficult to compute, e.g., we have
  $\Z_{29}^* / \langle 7 \rangle \times X \modsim \to X \modsim$ where
  $[x] = \{ x, x^7, x^{20}, x^{24}, x^{23}, x^{16}, x^{25} \} \in X
  \modsim$.
\end{exa}


\section{Conclusion}

Cryptographic group actions or semigroup actions provide a framework
that encompasses both the classical discrete logarithm problem as well
as interesting proposals for post-quantum cryptography.  In this
article we have examined this framework from a theoretical viewpoint
and studied the semigroup action problem as an analog of the discrete
logarithm problem.

In the case of group actions, the generic complexity can be considered
well understood, as there is both a square-root lower bound and a
square-root upper bound for the group action problem.  On the other
hand, for proper semigroup actions the situation appears to be less
clear.  The generic lower bound may exceed the square-root barrier,
however such instances tend to be degenerate and not interesting for
cryptography applications.

We also have discussed the potential of certain substructures to
weaken the hardness of the semigroup action problem, in particular in
the presence of non-units.  While such substructures do not guarantee
to practically break the semigroup action problem, they should be
taken into account when designing cryptosystems based on semigroup or
group actions.



\begin{thebibliography}{cc}
  \scriptsize\itemsep-.1mm

\bibitem{ADMP20} N.~Alamati, L.~De Feo, H.~Montgomery,
  S.~Patranabis, “Cryptographic group actions and applications,”
  in: Advances in Cryptology---ASIACRYPT 2020, pp.~411--439,
  LNCS 12492, Springer, 2020
  
\bibitem{Aru+19} F.~Arute, et al., “Quantum supremacy using a
  programmable superconducting processor,” \emph{Nature}~574 (2019),
  505--510

\bibitem{Cas+18} W.~Castryck, T.~Lange, C.~Martindale, L.~Panny,
  J.~Renes, “CSIDH: an efficient post-quantum commutative group action,”
  in: Advances in Cryptology---ASIACRYPT 2018, pp.~395--427, LNCS 11274,
  Springer, 2018

\bibitem{CD22} W.~Castryck, T.~Decru, “An efficient key recovery attack
  on SIDH,” 15 pages, IACR eprint 2022/975
  
\bibitem{Cou06} J.\,M.~Couveignes, “Hard homogenous spaces,”
  15 pages, IACR eprint 2006/291 (1997, published 2006)
  
\bibitem{DH76} W.~Diffie, M.\,E.~Hellman, “New directions in
  cryptography,” \emph{IEEE Trans.\ Inform.\ Theory}~22 (1976) 644--654

\bibitem{Gni14} O.\,W.~Gnilke, “The semigroup action problem
  in cryptography,” Ph.\,D.\ dissertation, University College Dublin
  (2014)

\bibitem{GJ21} R.~Granger, A.~Joux, “Computing discrete
  logarithms,” in: Computational Cryptography---Algorithmic Aspects of
  Cryptology, pp.~106--139, Cambridge Univ.\ Press, 2021

\bibitem{GKZ18} R.~Granger, T.~Kleinjung, J.~Zumbrägel,
  “Indiscreet logarithms in finite fields of small characteristic,”
  \emph{Adv.\ Math.\ Commun.}~12 (2018) 263--286

\bibitem{DJP14} L.~De Feo, D.~Jao, J.~Plût, “Towards
  quantum-resistant cryptosystems from supersingular elliptic curve
  isogenies,” \emph{J.\ Math.\ Cryptol.}~8 (2014) 209--247

\bibitem{GPSV21} S.~Galbraith, L.~Panny, B.~Smith,
  F.~Vercauteren, “Quantum equivalence of the DLP and CDHP for group
  actions,” \emph{Math.\ Cryptol.}~1 (2021) 40--44
  
\bibitem{Kob87} N.~Koblitz, “Elliptic curve cryptosystems,”
  \emph{Math.\ Comp.}~48 (1987) 203--209

\bibitem{Mau94} U.\,M.~Maurer, “Towards the equivalence of
  breaking the Diffie-Hellman protocol and computing discrete
  logarithms,” in: Advances in Cryptology---CRYPTO '94, pp.~271--281,
  LNCS 839, Springer, 1994

\bibitem{Mau05} U.\,M.~Maurer, “Abstract models of computation
  in cryptography,” in: Cryptography and Coding 2005, pp.~1--12,
  LNCS 3796, Springer, 2005
  
\bibitem{MW99} U.\,M.~Maurer, S.~Wolf, “The relationship between
  breaking the Diffie-Hellman protocol and computing discrete logarithms,”
  \emph{SIAM J.\ Comput.}~28 (1999) 1689--1721

\bibitem{Maz03} G.~Maze, “Algebraic methods for constructing
  one-way trapdoor functions,” Ph.\,D.\ dissertation, University of
  Notre Dame (2003)
  
\bibitem{MMR07} G.~Maze, C.~Monico, J.~Rosenthal, “Public-key
  cryptography based on semigroup actions,” \emph{Adv.\ Math.\
    Commun.}~1 (2007) 489--507
  
\bibitem{Mil86} V.\,S.~Miller, “Use of elliptic curves in
  cryptography,” in: Advances in Cryptology---CRYPTO '85, pp.~417--426,
  LNCS 218, Springer, 1986

\bibitem{Mon02} C.~Monico, “Semirings and semigroup actions in
  public-key cryptography,” Ph.\,D.\ dissertation, University of Notre
  Dame (2002)

\bibitem{MZ22} H.~Montgomery, M.~Zhandry, “Full quantum equivalence of
  group action DLog and CDH, and more,” in: Advances in Cryptology---%
  ASIACRYPT 2022 (to appear), Springer, 2022
  
\bibitem{PH78} S.\,C.~Pohlig, M.\,E.~Hellman, “An improved
  algorithm for computing logarithms over GF(p) and its cryptographic
  significance,” \emph{IEEE Trans.\ Inform.\ Theory}~24 (1978) 106--110
  
\bibitem{RS06} A.~Rostovtsev, A.~Stolbunov, “Public-key
  cryptosystem based on isogenies,” 19 pages, IACR eprint 2006/145
  
\bibitem{Sho97a} V.~Shoup, “Lower bounds for discrete logarithms and
  related problems,” in: Advances in Cryptology---EUROCRYPT '97,
  pp.~256--266, LNCS 1233, Springer, 1997

\bibitem{Sho97b} P.\,W.~Shor, “Polynomial-time algorithms for
  prime factorization and discrete logarithms on a quantum computer,”
  \emph{SIAM J.~Computing}~26 (1997) 1484--1509

\bibitem{Smi18} B.~Smith, “Pre- and post-quantum Diffie-Hellman
  from groups, actions, and isogenies,” in: Workshop on the Arithmetic
  of Finite Fields 2018, pp.~3--40, LNCS 11321, Springer, 2018
  
\bibitem{Sto10} A.~Stolbunov, “Constructing public-key
  cryptographic schemes based on class group action on a set of
  isogenous elliptic curves,” \emph{Adv.\ Math.\ Commun.}~4 (2010)
  215--235

\bibitem{Zum08} J.~Zumbrägel, “Public-key cryptography based on
  simple semirings,” Ph.\,D.\ dissertation, University of Zurich (2008)
  
\end{thebibliography}
\end{document}